\documentclass[3p,times]{elsarticle}
\makeatletter
\def\ps@pprintTitle{%
    \let\@oddhead\@empty
    \let\@evenhead\@empty
    \let\@oddfoot\@empty
    \let\@evenfoot\@empty}
\makeatother
\usepackage{amssymb}
\usepackage{amsthm}
\usepackage{amsmath,amssymb,amsfonts}
\usepackage{graphicx}
\usepackage{textcomp}
\usepackage{xcolor}
\usepackage{subcaption} % for subfigures
\usepackage{booktabs} % add this in the preamble
\usepackage{natbib}

\newcommand{\bea}{\begin{equation} \begin{aligned}}
\newcommand{\eea}{ \end{aligned}\end{equation}}

\newcommand{\beas}{\begin{equation*} \begin{aligned}}
\newcommand{\eeas}{ \end{aligned}\end{equation*}}

\newcommand{\beq}{\begin{equation}}
\newcommand{\eeq}{\end{equation}}
\newcommand{\beqs}{\begin{equation*}}
\newcommand{\eeqs}{\end{equation*}}

\newcommand{\bpm}{\begin{pmatrix}}
\newcommand{\epm}{\end{pmatrix}}
\newcommand{\bbm}{\begin{bmatrix}}
\newcommand{\ebm}{\end{bmatrix}}

\newtheorem*{standing_assumption}{Standing Assumption}
\newtheorem{theorem}{Theorem}
\newtheorem{corollary}{Corollary}
\newtheorem{lemma}{Lemma}
\newtheorem{remark}{Remark}
\newtheorem{assumption}{Assumption}

\begin{document}

\setlength{\bibsep}{0pt plus 0.3ex}

\begin{frontmatter}



\title{Consensus approximation and impulsive control\\ for a class of uncertain multi-agent dynamics}

\author[utcn]{Zolt\'an Nagy}
\ead{Zoltan.Nagy@aut.utcluj.ro}

\author[utcn,rac]{Lucian Busoniu}
\ead{Lucian.Busoniu@aut.utcluj.ro}

\author[utcn,cran]{Irinel Constantin Morărescu}
\ead{Constantin.Morarescu@univ-lorraine.fr}

\affiliation[utcn]{
  organization={Department of Automation, Technical University of Cluj Napoca},  
  city={Cluj Napoca},  
  country={Romania}
}

\affiliation[rac]{
  organization={Correspondent Member of the Romanian Academy},  
  city={Bucharest},  
  country={Romania}
}

\affiliation[cran]{
  organization={Université de Lorraine, CNRS, CRAN, F-54000},  
  city={Nancy},  
  country={France}
}

%% Abstract
\begin{abstract}
This paper studies a class of consensus dynamics where the interactions between agents are affected by a time-varying unknown scaling factor. This situation is encountered in the control of robotic fleets over a wireless network or in opinion dynamics where the confidence given to the peers varies in time. Firstly, we establish conditions under which practical upper and lower bounds on the consensus value can be determined.  Secondly, we propose control strategies for allocating a given control budget to shift agent states towards a desired consensus value despite the uncertainty. We provide computationally efficient linear programming-based approaches for both problems and validate the obtained results in numerical simulations.
\end{abstract}


% \begin{highlights}
% \end{highlights}

%% Keywords

\end{frontmatter}



\section{Introduction}
\label{sec:introduction}

Estimating and controlling the consensus value in multi-agent systems is a critical task with broad applications in fields such as marketing \cite{varma2018marketing, moruarescu2020space}, social influence \cite{cho2018dynamics, quattrociocchi2014opinion}, and distributed control \cite{mirtabatabaei2012opinion, wang2021opinion}. The interactions between the agents play a  major role in the overall asymptotic behavior but also in the mathematical analysis of the resulting model. Many of the proposed mathematical models come either from sociology \cite{degroot1974reaching,rainer2002opinion,Friedkin,martins2008continuous} or statistical physics \cite{clifford1973model, sznajd2000opinion}. Some of them lead to consensus, while others generate polarization within the network. In the simple case where the interactions between agents are linear and the interaction network is fixed and connected, the consensus value can be found as the inner product between the initial opinion vector and the left eigenvector associated with the zero eigenvalue of the network’s Laplacian matrix, see for instance \cite{moruarescu2020space}. However, in the general case where the interactions are time-varying or nonlinear the exact consensus value is rarely representable in closed form. %To define the interactions between agents in opinion dynamics models, several approaches can be considered. For instance, in \cite{martins2008continuous}, agent interactions are modeled using the Voter model \cite{clifford1973model}, and the Sznajd model \cite{sznajd2000opinion}, both of which consider interactions between pairs of agents. 
%\textcolor{blue}{To gain better control over the structure of interactions and achieve a more realistic model, agent interactions can be defined over a network. Several models have been proposed in the literature for such interactions, including the DeGroot model \cite{degroot1974reaching} and the Hegselmann–Krause model \cite{rainer2002opinion}.} In the case where the opinion update protocol is linear, the consensus value can be found as the inner product between the initial opinion vector and the left eigenvector associated with the zero eigenvalue of the network’s Laplacian matrix \cite{moruarescu2020space}.In the nonlinear case, however, the situation is more complex. 
For example, in \cite{olfati2004consensus}, consensus estimation is addressed for balanced directed graphs. Due to the balanced structure, the consensus value in that case can be computed as the average of the initial values, a property that holds even under switching topologies provided the graph remains balanced at all times. A more general framework is introduced in \cite{ren2005consensus}, where it is shown that, even with switching topologies, a consensus value exists and the agent opinions will converge to it, as long as the switching is sufficiently frequent and the union of the graphs over time contains a spanning tree. The exact value of the consensus remains, however, generally intractable \cite{ren2005consensus}.

% \color{blue}
While computing the exact consensus value is generally not feasible, it is possible to determine a range within which this value lies. However, existing results in the literature typically offer very conservative bounds. For instance, in \cite{ren2005consensus, chowdhury2016continuous}, the consensus value is bounded between the minimum and maximum of the agents’ initial states, which is 
uninformative in many practical cases.
A more refined approach is presented in \cite{martin2019consensus}, where time-varying interactions are considered for continuous-time opinion dynamics, and a less restrictive bounding method is proposed. Nevertheless, the method has numerical limitations hampering its application in large networks.

The control of the consensus value has previously been addressed in the context of opinion dynamics \cite{varma2018marketing, moruarescu2020space, etesami2021open, jond2024differential, alkhorshid2024saturated}. In \cite{varma2018marketing}, a game-theoretic framework is employed to identify the Nash equilibrium in a duopoly setting, where competing marketers seek optimal strategies to sway public opinion. In \cite{etesami2021open}, competitive influence maximization under the DeGroot model \cite{degroot1974reaching} is analyzed, and it is shown that optimal budget allocations can be computed efficiently. In \cite{alkhorshid2024saturated}, a decentralized control strategy is developed for multi-agent systems with state and input constraints, where Lyapunov-based methods are used to design feedback laws.
Reference \cite{moruarescu2020space} focuses on the optimal allocation of a marketing budget to steer agents’ opinions toward a desired target.
% While most papers on opinion influence, see e.g. \cite{varma2018marketing,etesami2021open,jond2024differential,alkhorshid2024saturated}, model the control input as affecting the interaction dynamics, only a few works consider direct influence on the agents’ opinions (e.g., \cite{moruarescu2020space}). 
All these works assume complete model knowledge, while some -- like \cite{moruarescu2020space} -- further restrict to linear dynamics.
% they typically assume linear consensus dynamics or rely on dynamics that are fully known. 
Such assumptions fail to capture the inherent uncertainties present in real-world interactions. 

To address the conservativeness of existing consensus bounds as well as the lack of consensus control for uncertain dynamics, we consider multi-agent systems with time-varying uncertain interactions and propose consensus approximation results and a consensus-influencing control strategy for such systems.
In the cooperative robotics context, our model can handle variations of the communication strength in wireless networks \cite{loheac2022time}. In social networks, our model is an extension of \cite{chowdhury2016continuous}, with an important generalization: instead of incorporating a single type of stubbornness, we introduce uncertainty into the agents' dynamics, covering a much larger class of opinion updates. More precisely, in our model, for each agent, the usual linear influence of the neighbors is multiplied by a bounded, uncertain term, reflecting fluctuations over time in individuals’ susceptibility to being influenced by others. 
% External influences on agents’ opinions are also considered in other works; for example, \cite{couthures2024analysis} studies an opinion dynamics model coupled with external environmental dynamics. 

Our first key contribution is an approach to estimate the consensus value by computing less conservative lower and upper bounds than in \cite{ren2005consensus,chowdhury2016continuous}. The problem of computing the bounds is formulated as nonlinear optimization and then converted into a linear program (LP) for an efficient solution.
In the second key contribution, we exploit our bounds to derive an analytically grounded method to influence agent opinions. Based on \cite{moruarescu2020space}, we combine the discrete-time dynamics of the agents with marketing campaigns modeled as control inputs that impulsively influence the agents' initial opinions toward a target value. The campaign budget is limited: there is an upper bound on the total budget, as well as an upper bound on the maximum input that can be applied to any individual agent. 
%We consider a setup in which the control input impulsively alters the underlying discrete-time dynamics of the agents.
Our objective in allocating the marketing budget is to minimize the difference between the consensus value and the target value. Since the consensus value cannot be directly determined, we approximate it by our earlier determined lower or upper bounds, and the budget is then allocated so as to minimize the difference between these bounds and the target. A nonlinear optimization problem results, which we again relax to an LP problem, to allow efficient computation. To validate both the consensus bounds and the budget allocation strategy, we conduct experiments on scale-free networks, which are commonly observed in real-world social systems \cite{barabasi1999emergence}.
For comparative analysis, we use the allocation method from \cite{moruarescu2020space}, which achieves optimal budget allocation in the linear case but provides no guarantees in this nonlinear~setting.% We use this method as a baseline.

% We obtain comparative results and we highlight that i and then we use it to compare with our results.  Second, we enhance this allocation strategy by transforming the original nonlinear problem into a linear programming formulation to enable efficient~computation.

% In our second key contribution, we combine the marketing budget allocation strategy with the bound approximation, leading to a nonlinear optimization problem. Then, first we consider a baseline method of allocation from \cite{moruarescu2020space}, which allocates the budget optimally in the linear case, and second we further improve this allocation strategy and transform the original nonlinear optimization problem into a linear programming version to enable efficient computation.

% While in most of the opinion influence papers in the related literature, see e.g. \cite{varma2018marketing, etesami2021open, jond2024differential, alkhorshid2024saturated} the control input is directly influencing the dynamics of the agents, and only a few of these, see e.g. \cite{moruarescu2020space} considers the direct inlfuence of the agents opinion. Therefore this is considered as an interesting research gap, and in contrast to \cite{moruarescu2020space}, where linear dynamics is considered, our work considers uncertain nonlinear dynamics and focuses on providing theoretical guarantees on the distance between the target value and the consensus value.

\color{black}
The paper is organized as follows. Section~\ref{sec:problem_statement} introduces the model and outlines the consensus value approximation problem. Section~\ref{sec:consensus_bounds} presents the design and analysis of consensus bounds, along with the proposed optimization framework and supporting experimental results. Section~\ref{sec:consensus_control} addresses the control problem of marketing budget allocation and illustrates the effectiveness of the approach through various scenarios. Section~\ref{sec:conclusions} concludes the paper and outlines potential directions for future research.

% Utilizing such network structures allows for a more nuanced analysis of opinion dynamics and offers deeper insights into how individual interactions shape collective outcomes.

% In linear models, the consensus value can be analytically derived with relative ease. However, in the more general and realistic case of nonlinear opinion dynamics, determining the consensus value becomes significantly more challenging. This complexity increases further when incorporating uncertainty into the model, which better reflects the variability and unpredictability found in real-world systems.

% Having information about the consensus value is important knowledge that could be used to generate control action based on them. For instance, in the case of market share competition between two competitors, predicting the future profitability of each competitor allows for targeted marketing campaigns to influence the consensus outcome and shape the market dynamics. This paper presents a practical approach to formally determine the consensus value and establish lower and upper bounds. 

% Furthermore, we develop a set of conditions that can be used to allocate a marketing budget to maximize the achievable consensus value.

\textbf{Notations:} The symbol $I$ denotes the identity matrix of appropriate dimension, while  $\mathbf{1}$ and $\mathbf{0}$ refer to column vectors of ones and zeros, respectively. 
%The matrix $0_{N\times M}$ refers to an $N$ by $M$ matrix with $0$ elements. 
Notation $\phi$ refers to a column vector with elements $\phi = [\phi_1\,\, \phi_2 \,\, ...\,\, \phi_n]^T$, and we use $\mathrm{diag}(\phi)$ to refer to a diagonal matrix with  elements $\phi_i$, $i=1,...,n$ on the diagonal.

%% Use \subsection commands to start a subsection.
\section{Problem statement}
\label{sec:problem_statement}

Consider a set $V=\{1,...,n\}$ of $n$ agents interacting over a network described by a directed graph given by a weighted adjacency matrix $A$, whose components $a_{ij}\in [0, \, 1]$ $\forall$ $i,\,j\in V$ refers to the nominal connection strength from agent $i$ to agent $j$. We assume that $a_{ii}=0,\ \forall i\in V$ and consider the following agent dynamics:
\bea
x_i(k+1) = x_i(k)+\gamma_i(k)\sum_{j=1}^n a_{ij}(x_j(k)-x_i(k))
\label{eq:coca_model}
\eea
where scalar $x_i(k)$ refers to the opinion (state) of agent $i$ at discrete time step $k$. The terms $\gamma_i(\cdot)$ for $i\in V$ denote the time-varying uncertainties on the interaction strength, which are assumed to be uniformly bounded:
\begin{equation}
%\frac{\underline{\omega}}{n_i} = 
\underline{\gamma}_i \leq \gamma_i(\cdot) \leq \overline{\gamma}_i, 
%= \frac{\overline{\omega}}{n_i},
\quad \forall \,i\in V
\label{eq:bounded_nonlin}
\end{equation}
where $\underline{\gamma}_i=\frac{\underline{\omega}}{n_i}$ and $\overline{\gamma}_i = \frac{\overline{\omega}}{n_i}$, with $n_i$ denoting the number of neighbors of agent $i$ (number of nonzero values $a_{ij}$, with $j\in V$), while $0 < \underline{\omega} \leq \overline{\omega} \leq 1$ are given~constants. The scaling factor is divided by $n_i$, so that the effect on each agent is normalized by its number of neighbors.
% We denote by $\underline{\omega}$ and $\overline{\omega}$ the global uncertainty bounds satisfying . 
%Based on these, the local uncertainty bounds for agent $i$ are $\underline{\gamma}_i = \frac{\underline{\omega}}{n_i}$ and $\overline{\gamma}_i = \frac{\overline{\omega}}{n_i}$, where $n_i$ represents the number of neighbors of agent $i$.
%Based on these, the local uncertainty bounds for agent $i$ are $\underline{\gamma}_i = \frac{\underline{\omega}}{n_i}$ and $\overline{\gamma}_i = \frac{\overline{\omega}}{n_i}$, where $n_i$ represents the number of neighbors of agent $i$.
Using the notation $x(k) = [x_1(k),\, x_2(k)\,,...,\, x_n(k)]^T$ and $\gamma(k) = [\gamma_1(k), \, \gamma_2(k),\,...\,,\, \gamma_n(k)]^T$ we rewrite \eqref{eq:coca_model} as follows:
\bea
x(k+1) &= (I+\mathrm{diag}\big(\gamma(k)\big) L) x(k)
% D(k) &= (I+\mathrm{diag}\big(\gamma(k)\big) L) 
\label{eq:coca_model_matrix_form}
\eea
where $L$ represents the Laplacian matrix of the graph, with $l_{ij} = -a_{ij}$ for $i\neq j$ and $l_{ii}=\sum_{i\neq j} a_{ij}$. 
% and $\Gamma(k)$ has the form:
% \bea
% \Gamma(k) = &\bbm
%     \gamma_1(k) & 0 & ... & 0\\
%      & \gamma_2 (k) &  ... & 0\\
%     ... & ... & ... & ...\\
%     0 & 0 &  ... & \gamma_n(k)\\
%     \ebm,\\
% \eea

An example of dynamics \eqref{eq:coca_model_matrix_form} is the continuous opinions - continuous actions (COCA) model, presented in \cite{chowdhury2016continuous}. There, the nonlinear function $\gamma_i(\cdot)$ depends on the opinion of the agent and is called stubbornness:
\bea
\label{eq:COCA_stubborness}
\gamma_i(x_i(k)) =& \frac{x_i(k)(1-x_i(k))}{n_i}
\eea
In this case, we have $\underline{\omega} = 0$ and $\overline{\omega} = 0.25$. Note that $\overline{\omega}$ is obtained when $n_i = 1$ and the agent's opinion is $ 0.5$. Furthermore, $\underline{\omega} = 0$ occurs when the opinion of agent $i$ is either $0$ or $1$. In such extreme cases, the agent's opinion remains fixed and is not influenced by any of its neighbors. To avoid these scenarios, we follow \cite{chowdhury2016continuous} by requiring that the constant $\underline{\omega} > 0$.

% \subsection{Problem description}

For the model to achieve consensus, we impose the following standing assumption throughout this study.

\begin{standing_assumption}
\label{assu:strongly_connected}
The network of agents is strongly connected, implying that the Laplacian matrix $L$ has a single $0$ eigenvalue and that all the other eigenvalues are positive.
\end{standing_assumption}
Based on this, we can formulate the following lemma.

%We first start by proving that if the $L$ matrix represents a connected network of agents than the opinion dynamics in \eqref{eq:coca_model_matrix_form} converges to a consensus.

\begin{lemma}
\label{prop:conv_to_cons}
The nonlinear opinion dynamics \eqref{eq:coca_model_matrix_form} asymptotically converge to a consensus value, which we denote by $\alpha$.
\end{lemma}

\begin{proof}
Lemma 2.1 in \cite{ren2005consensus} states that for time-varying switching topology dynamics, the system converges to a consensus value $\alpha$ if the union of the directed interaction graphs over certain time intervals contains a spanning tree.
The network in \eqref{eq:coca_model_matrix_form} is and remains strongly connected because $\mathrm{diag}(\gamma(k))$ never sets the connection weights to $0$ due to \eqref{eq:bounded_nonlin}. This ensures a spanning tree is always present.
%, so  \eqref{eq:coca_model_matrix_form} converges to a consensus value $\alpha$.
\end{proof}

In the nonlinear case, as noted in~\cite{ren2005consensus}, the consensus value generally cannot be determined in closed form. Nonetheless, it is possible to establish practical bounds. We denote lower and upper bounds of \(\alpha\) by \(\alpha_{\min}\) and \(\alpha_{\max}\), respectively:
\bea
\alpha_{\min}\leq \alpha \leq \alpha_{\max}
\label{eq:alpha_min_max_basic}
\eea
Previous works \cite{olfati2004consensus}-\cite{chowdhury2016continuous} have presented highly conservative solutions, where the bounds are given by  

\bea
\alpha_{\min} = \min(x(0)), \quad \alpha_{\max} = \max(x(0))
\eea

However, such bounds provide limited practical insight. A key contribution of our present work lies in the development of a less restrictive method to determine lower and upper bounds for consensus, given next.
%This will be presented in the next section.

\section{Consensus bounds}
\label{sec:consensus_bounds}

Next, Section \ref{subsec:bounds_analytical} presents some essential analytical results, Section \ref{sec:Optimization} provides a computationally efficient solution to find consensus bounds in \eqref{eq:alpha_min_max_basic}, while Section \ref{sec:assump_differentiate} presents a set of examples illustrating the characterization power of these bounds.
%, while the fourth part gives a comparison with existing results.

\subsection{Analytical results}
\label{subsec:bounds_analytical}
For strongly-connected networks with linear opinion dynamics, the consensus value is given by $\alpha = \nu^Tx(0)$, where $\nu$ is the left eigenvector corresponding to the eigenvalue \(0\) of the Laplacian $L$ \cite{ren2005consensus}. For the model described by~\eqref{eq:coca_model_matrix_form}, such a closed-form solution is not available due especially to the time variation of the eigenvector $\nu$.
% and depends on the evolution of the opinions.
Let us point out that the normalized left-eigenvector corresponding to the zero eigenvalue at any given step $k$ can still be explicitly characterized.
For any $\gamma(k)$, consider the matrix $\mathrm{diag}(\gamma(k))L$; a zero eigenvalue always exists at any step \( k \) due to the strong connectivity of the network. We denote the normalized eigenvector corresponding to the $0$ eigenvalue by $\nu_{\gamma(k)}$, so that $\nu_{\gamma(k)}^T\mathbf{1}=1$. 
By definition of the eigenvector, we have:
% We start by studying the unnormalized version denoted with $\tilde\nu_{\gamma(k)}$:

\bea
\underbrace{\nu_{\gamma(k)}}_{\text{eigenvector of $\mathrm{diag}(\gamma(k)) L$}}^T\mathrm{diag}(\gamma(k)) L = \mathbf{0}^T
\label{eq:eigen_vector_prop_COCA}
\eea
Since $\mathrm{diag}(\gamma(k))$ is diagonal, $\nu_{\gamma(k)}^T\mathrm{diag}(\gamma(k))$ is  left eigenvector corresponding to eigenvalue $0$ of $L$, so 

\beqs
\underbrace{\nu_{\gamma(k)}^T\mathrm{diag}(\gamma(k))}_{\text{eigenvector of $L$}} L = \mathbf{0}^T
\eeqs
% \bea
% \bbm \nu_{1\,\gamma_1(k)}\gamma_1(k) & ... &\nu_{n\,\gamma_n(k)}\gamma_n(k) \ebm^T L = \mathbf{0}^T
% \label{eq:eigen_vector_prop_COCA_ext}
% \eea
% where $\nu_{i\,\gamma_i(k)}$ refers to the $i$-th element of $\nu_{\gamma(k)}$. We can observe that $\nu^T_{\gamma(k)}\mathrm{diag}(\gamma(k))$ is a left eigenvector corresponding to eigenvalue $0$ of $L$. 
Knowing that $L$ has a unique normalized left eigenvector $\nu$ corresponding to eigenvalue $0$, one gets that the vector
$\nu^T_{\gamma(k)}\mathrm{diag}(\gamma(k))$ is a scalar multiple of $\nu^T$ which writes as
% Note that, $\nu^T_{\gamma(k)}\mathrm{diag}(\gamma(k))\mathbf{1}\neq 1$, therefore we  normalize as follows:
% \bea
% \underbrace{\frac{1}{\rho(k)}\nu_{\gamma(k)}^T\mathrm{diag}(\gamma(k))}_{\text{normalized eigenvector of $L$}} L =\mathbf{0}^T
% \label{eq:eigen_vector_prop_COCA_ext_norm_one}
% \eea
% With $\rho(k)$ denoting a constant term for each $k$. 
%that normlizes $\nu^T_{\gamma(k)}\mathrm{diag}(\gamma(k))$.
% Since there is a unique normalized left eigenvector corresponding to eigenvalue $0$,
% we consider equations \eqref{eq:eigen_vector_prop} and \eqref{eq:eigen_vector_prop_COCA_ext_norm_one}, from which we have:
\beqs
% \bbm \nu_{1\,\gamma_1(k)}\gamma_1(k) & ... &\nu_{n\gamma_n(k)}\gamma_n(k) \ebm^T  
\nu_{\gamma(k)}^T\mathrm{diag}(\gamma(k))
= q(k)\nu^T
\eeqs
where $q(k)$ denotes the scaling term for each $k$. We examine $\nu_{i\,\gamma_i(k)}$, the $i$-th element of $\nu_{\gamma(k)}$:
\bea
\nu_{i\,\gamma_i(k)}\gamma_i(k)=q(k)\nu_i, \quad \forall i\in V 
\label{eq:nu_i}
\eea
% where $\nu_{i\,\gamma_i(k)}$ refers to the .
Since $\gamma_i(k)>0$ for all $i=1,...,n$, one obtains
\beqs
    \nu_{i\,\gamma_i(k)} = \frac{\nu_i}{\gamma_i(k)}q(k), \quad \forall i\in V
\eeqs
% or in matrix form:
% \bea
% \nu_{\gamma(k)}^T = \nu^T\mathrm{diag}\Big(\frac{1}{\gamma(k)}\Big) \alpha(k)
% \eea
The normalization $\nu_{\gamma(k)}^T\mathbf{1} = 1$ yields
\beqs
1 = \sum_{i=1}^n \nu_{i\gamma_i(k)} = \sum_{i=1}^n  \frac{\nu_i}{\gamma_i(k)}q(k)
\eeqs
leading to the expression
\bea
q(k) = \frac{1}{\sum_{i=1}^n\frac{\nu_i}{\gamma_i(k)}} = \frac{1}{\nu^T\mathrm{diag}\Big(\frac{1}{\gamma(k)}\Big)\textbf{1}}
\label{eq:nu_almost_done}
\eea
Finally, from \eqref{eq:nu_i} and \eqref{eq:nu_almost_done}, we obtain:
\bea
\nu_{\gamma(k)}^T =  \frac{\nu^T\mathrm{diag}\Big(\frac{1}{\gamma(k)}\Big)}{\nu^T\mathrm{diag}\Big(\frac{1}{\gamma(k)}\Big)\mathbf{1}}
\label{eq:left_eigen_v}
\eea
In what follows, we introduce the following two auxiliary dynamics:
\bea
\underline{x}(k+1) =& (I-\mathrm{diag}(\gamma_*)L)\underline{x}(k)
\label{eq:coca_lower_dynamics}
\eea
\bea
\overline{x}(k+1) =& (I-\mathrm{diag}(\gamma^*) L)\overline{x}(k)
\label{eq:coca_upper_dynamics}
\eea
where $\underline{x}(0) = \overline{x}(0)= x(0)$, and the terms \(\gamma^*\) and \(\gamma_*\) are taken such that:
\bea
\underline{\nu}^Tx(0) = \frac{\nu^T\mathrm{diag}\Big(\frac{1}{\gamma_*}\Big)}{{\nu^T\mathrm{diag}\Big(\frac{1}{\gamma_*}\Big)\mathbf{1}} } x(0) = \underset{\gamma}{\min}\, \bpm \frac{\nu^T\mathrm{diag}\Big(\frac{1}{\gamma}\Big)}{{\nu^T\mathrm{diag}\Big(\frac{1}{\gamma}\Big)\mathbf{1}} } x(0)\epm\\
\\
\overline{\nu}^Tx(0) = \frac{\nu^T\mathrm{diag}\Big(\frac{1}{\gamma^*}\Big)}{{\nu^T\mathrm{diag}\Big(\frac{1}{\gamma^*}\Big)\mathbf{1}} }x(0) = \underset{\gamma}{\max}\, \bpm \frac{\nu^T \mathrm{diag}\Big(\frac{1}{\gamma}\Big)}{{\nu^T\mathrm{diag}\Big(\frac{1}{\gamma}\Big)\mathbf{1}} }x(0)\epm
\label{eq:optimization_problem}
\eea
In the $\min$ and $\max$ arguments, $\gamma_i\in[\underline{\gamma}_i,\, \overline{\gamma}_i]$ for $i\in V$.   
%Here $\gamma_i\in[\underline{\gamma}_i, \, \overline{\gamma}_i]$, for all $i=1,..,n$. The terms $\mathrm{diag}\Big(\frac{1}{\gamma_*}\Big)$ and $\mathrm{diag}\Big(\frac{1}{\gamma^*}\Big)$ refer to the optimistic and pessimistic consensus values achievable from $x(0)$.
Thus, $\underline{\nu}$ and $\overline{\nu}$ represent the normalized left eigenvectors corresponding to eigenvalue $0$ of $\mathrm{diag}(\gamma_*)L$ and  $\mathrm{diag}(\gamma^*)L$ respectively. \\
%Finding $\underline\nu$ and $\overline{\nu}$ in \eqref{eq:optimization_problem} is a linear programming (LP) problem that can be efficiently solved.
% Based on Lemma \ref{prop:left_eigen_vector} we have:
% \bea
% \underline\nu^T = \frac{\nu^T\mathrm{diag}\Big(\frac{1}{\gamma_*}\Big)}{{\nu^T\mathrm{diag}\Big(\frac{1}{\gamma_*}\Big)\mathbf{1}} }, 
% \quad 
% \overline\nu^T = \frac{\nu^T\mathrm{diag}\Big(\frac{1}{\gamma^*}\Big)}{{\nu^T\mathrm{diag}\Big(\frac{1}{\gamma^*}\Big)\mathbf{1}} },
% \eea
%Finding them is an optimization problem, and we discuss it in Section \ref{sec:Optimization}.
%$\mathrm{diag}\Big(\frac{1}{\gamma_*}\Big)$ and $\mathrm{diag}\Big(\frac{1}{\gamma^*}\Big)$ refer to the optimistic and pessimistic consensus values achievable from $x(0)$.
%The dynamics in \eqref{eq:coca_lower_dynamics} and \eqref{eq:coca_upper_dynamics} are used to date
%The consensus values of \eqref{eq:coca_lower_dynamics} and \eqref{eq:coca_upper_dynamics} are given by $\nu_{\gamma_\star}^\top x(0)$ and $\nu_{\gamma^\star}^\top x(0)$, respectively, since the dynamics are linear.
We also introduce the following auxiliary variables:
\beqs
\underline\theta(k) := \underline\nu^T{p}(k), \quad
\overline\theta(k) := \overline\nu^T{p}(k)
\eeqs
% The terms \(\underline{\alpha}(k)\) and \(\overline{\alpha}(k)\) are constants, as \(\underline{p}(k)\) and \(\overline{p}(k)\) evolve according to the linear dynamics in \eqref{eq:coca_lower_dynamics} and \eqref{eq:coca_upper_dynamics}.
Since $\overline{\nu}^\top \mathbf{1} = \underline{\nu}^\top \mathbf{1} =1$, it follows that $\underline\theta(k)$ and $\overline\theta(k)$ are weighted averages of the agents' opinions at time $k$. Therefore, we have:
\bea
\label{eq:theta_under_over_converges_beta}
\lim_{k\rightarrow \infty} \underline\theta(k) =
\lim_{k\rightarrow \infty} \overline\theta(k) = \alpha
\eea
The setup under consideration is completed by the following condition.
\begin{assumption}
\label{ass:bounded_projection_condition}
For all steps \( k \):
\begin{align}
\underline{\nu}^T \, \mathrm{diag}\big(\gamma(k)\big) \, L \, x(k) &\leq 0 \label{eq:lower_deriv_condit} \\
\overline{\nu}^T \, \mathrm{diag}\big(\gamma(k)\big) \, L \, x(k) &\geq 0 \label{eq:upper_deriv_condit}
\end{align}
% These conditions impose upper and lower linear constraints on the term 
% \( \mathrm{diag}(\gamma(k)) L x(k) \), where \( \underline{\nu} \) and \( \overline{\nu} \) are given vectors.
\end{assumption}

%Assumption~\ref{ass:bounded_projection_condition} can be intuitively interpreted as stating that the evolution of $\underline{\theta}$ is monotonically increasing, while that of $\overline{\theta}$ is monotonically decreasing.
We illustrate later on in Section \ref{sec:assump_differentiate} that in most practical cases this assumption holds. We are now ready to formulate our first main result.
\begin{theorem}
\label{th:lower_upper_consensus}
Under Assumption \ref{ass:bounded_projection_condition}, the following lower and upper bounds on the consensus value hold:
\bea
\alpha_{\max} = \overline{\nu}^Tx(0),
\quad \alpha_{\min} = \underline{\nu}^Tx(0)
\label{eq:alpha_bounds}
\eea
\end{theorem}

\begin{proof}
Both $\underline\theta(k)$ and $\overline{\theta}(k)$ converge to $\alpha$ according to \eqref{eq:theta_under_over_converges_beta}. Therefore, we aim to prove that for $\forall k$:
\beqs
\underline{\theta}(k) \geq   \underline{\nu}^Tx(0), \quad
\overline{\theta}(k) \leq   \overline{\nu}^Tx(0)\\
\eeqs
%We prove this by induction. 
Dynamics \eqref{eq:coca_model_matrix_form}, \eqref{eq:coca_lower_dynamics}, and \eqref{eq:coca_upper_dynamics} all start from the same initial condition $\overline{x}(0)= \underline{x}(0)=x(0)$, therefore $\underline\theta(0)= \underline{\nu}^Tx(0)$ and $\overline\theta(0)=\overline{\nu}^Tx(0)$.
We reason by induction and we consider $\overline\theta(k)\leq \overline{\nu}^Tx(0)$ holds in order to prove that $\overline\theta(k+1)\leq \overline{\nu}^Tx(0)$ follows. Straightforward computation yields
\beas
\overline{\nu}^Tx(0) -\overline{\theta}(k+1) 
&=\overline{\nu}^T x(0) -  \overline{\nu}^T  (I-\mathrm{diag}(\gamma(k))L)x(k)=\overline{\nu}^T x(0) -  \overline{\nu}^Tx(k)+ \overline{\nu}^T\mathrm{diag}(\gamma(k))Lx(k)\\
&=\overline{\nu}^T x(0) -  \overline{\theta}(k)+ \overline{\nu}^T\mathrm{diag}(\gamma(k))Lx(k)
\eeas
By the induction hypothesis $\overline{\nu}^Tx(0)-\overline\theta(k)\geq 0$ and from  Assumption~\ref{ass:bounded_projection_condition} one has $\overline{\nu}^T\mathrm{diag}(\gamma(k))Lx(k)\geq 0$. Therefore, one obtains
\beqs
\overline{\nu}^Tx(0) -\overline{\theta}(k+1)\geq 0
\eeqs
Since $\alpha_{\max}=\overline{\nu}^Tx(0)$, we have proven that $\alpha_{\max}\geq \overline{\theta}(k)$ for $\forall k$.
Furthermore, $\lim_{k\rightarrow\infty}\overline{\theta}(k)=\alpha$, so $\alpha\leq \alpha_{\max}$.\\
Similarly we obtain $\alpha_{\min}\leq \alpha$.
\end{proof}
\begin{remark}
    From the proof of Theorem~\ref{th:lower_upper_consensus}, $\underline{\theta}(k)$ monotonically increases, while $\overline{\theta}(k)$ monotonically~decreases.
\end{remark}
% \begin{remark}
% Theorem~\ref{th:lower_upper_consensus} with Assumption~\ref{ass:bounded_projection_condition} provides a conservative approach sinc
% \end{remark}
\begin{remark}
\label{rem:limited_cons}
Rewriting
$\alpha = \delta \alpha_{\min} + (1-\delta)\alpha_{\max}$, for some $\delta\in [0,\,1]$, we get
    \beas
     \alpha  = &\delta\underline{\nu}^Tx(0) +(1-\delta) \overline{\nu}^Tx(0) =\delta\frac{\nu^T \mathrm{diag}\Big(\frac{1}{\gamma_*}\Big)}{\nu^T \mathrm{diag}\Big(\frac{1}{\gamma_*}\Big)\mathbf{1}}x(0)+ (1-\delta) \frac{\nu^T \mathrm{diag}\Big(\frac{1}{\gamma_*}\Big)}{\nu^T \mathrm{diag}\Big(\frac{1}{\gamma_*}\Big)\mathbf{1}}x(0)
    %= & \frac{\nu^T \Big(\delta\,\,\mathrm{diag}\Big(\frac{1}{\gamma_*}\Big) + (1-\delta)\mathrm{diag}\Big(\frac{1}{\gamma_*}\Big) \Big)}{\nu^T \Big(\delta\,\,\mathrm{diag}\Big(\frac{1}{\gamma^*}\Big) + (1-\delta)\mathrm{diag}\Big(\frac{1}{\gamma^*}\Big) \Big)\mathbf{1}}x(0)  \\
% =: & \delta\nu^T\diag()
=: \frac{\nu^T \mathrm{diag}\Big(\frac{1}{\gamma}\Big)}{\nu^T\mathrm{diag}\Big(\frac{1}{\gamma}\Big)\mathbf{1}}x(0)
\eeas
where $\displaystyle\frac{1}{\gamma} = \frac{\delta}{\nu^T\mathrm{diag}\big(\frac{1}{\gamma_*}\big)}\frac{1}{\gamma_*}+\frac{{1-\delta}}{\nu^T\mathrm{diag}\big(\frac{1}{\gamma^*}\big)}\frac{1}{\gamma^*} =\left [\frac{1}{\gamma_1},...,\frac{1}{\gamma_n}\right]^T $ is a constant vector with $\gamma_i\in [\underline{\gamma_i}, \, \overline{\gamma_i}]$ for $i=1,...,n$.
\end{remark}

In the sequel, we present an optimization problem to efficiently compute $\gamma_*$ and $\gamma^*$ of \eqref{eq:optimization_problem}, thereby obtaining the bounds \eqref{eq:alpha_bounds} on the consensus value.
%a set of experimental results to highlight the practical significance of Theorem~\ref{th:lower_upper_consensus}.

\subsection{Computing the bounds efficiently}
\label{sec:Optimization}

% In this section, we investigate the methodology for determining the lower and upper bounds for the consensus.
%The \textit{optimistic scenario} corresponds to identifying the maximum achievable consensus value originating from the initial condition \( x(0) \), given the known bounds of \( \gamma_i(\cdot) \) for all \( i = 1, \dots, n \). Conversely, the \textit{pessimistic scenario} involves determining the minimum attainable consensus value under the same conditions.

We introduce the change of variables $\phi_i := \frac{1}{\gamma_i}$, $\underline\phi_i := \frac{1}{\overline{\gamma}_i}$ and $\overline\phi_i := \frac{1}{\underline{\gamma}_i}$, allowing us to rewrite the two optimization problems in \eqref{eq:optimization_problem} as follows:

\bea 
&\underset{\phi}{\min}\, \frac{\sum_{i=1}^n\nu_i \phi_i x_i(0)}{\nu^T\phi} \\%=:\, C(\phi)
% \label{eq:opt_problem_proof_min}
% \eea 
% \bea 
&\underset{\phi}{\max}\, \frac{\sum_{i=1}^n\nu_i \phi_i x_i(0)}{\nu^T\phi}
\label{eq:opt_problem_proof}
\eea 

\bea
\mathrm{subject\,\, to:}\quad 
    % \phi_i&\leq \overline{\phi_i}, \\
    % -\phi_i&\leq -\underline{\phi_i}, \\
    \phi_i &\in [\underline{\phi_i},\, \overline{\phi_i}], \, \forall i=1,..,n \\
    %\frac{1}{\overline{\omega}}&\leq \underline{\phi_i}\leq \overline{\phi_i} <\frac{1}{\underline{\omega}}
    \label{eq:opt_constr_proof}
\eea
The constraints \eqref{eq:opt_constr_proof} can be reformulated in matrix form:
\bea
A\phi\leq b,\quad \mathrm{where }\quad 
A = \bbm 
    I\\
    -I
    \ebm,\quad
b = \bbm \,\,\,\overline{\phi} \,\,\,\\ \,\,\, \underline{\phi} \,\,\,\ebm
\eea

The optimization problems defined in \eqref{eq:opt_problem_proof} subject to the constraints in \eqref{eq:opt_constr_proof} is a linear fractional programming (LFP) problem. To reformulate it as a Linear Programming (LP) problem, we employ the Charnes-Cooper transformation~\cite{cooper1962programming}:
% \bea
% C(x) = \frac{c^Tx+\zeta_1}{d^Tx+\zeta_2}
% \eea
% subject to 
% \bea
%     Ax\leq b%\\
%     % x>0
% \eea
% We have the correspondence: 
% \bea
% x &:=\phi, \quad d := \nu, \quad \zeta_1:= 0, \quad \zeta_2 := 0,\\
% c &:= [\nu_1p_1(0),\,\nu_2p_2(0),\,...,\,\nu_np_n(0)]^T
% \eea
% Now based on \cite{cooper1962programming}, we define:
\beas
\tau := \frac{1}{\nu^T\phi},\quad \chi := \phi\tau
\eeas
By denoting $ c:=[\nu_1\phi_1,\,...\,,\nu_n\phi_n]^T$, \eqref{eq:opt_problem_proof} and \eqref{eq:opt_constr_proof} are rewritten as:
\begin{align}
& \underset{\chi}{\min} \quad c^T \chi \label{eq:min_opt} \\
& \underset{\chi}{\max} \quad c^T \chi \label{eq:max_opt}
\end{align}
\bea
\mathrm{both \,\,subject\,\, to:}\quad 
A\chi& \leq b\tau, \quad 
\nu^T\chi& = 1,\quad
\tau&>0
\label{eq:opt_constraints}
\eea
We denote the optimal solution of \eqref{eq:min_opt} by \( \tau^* \) and \( \chi^* \), and that of \eqref{eq:max_opt} by \( \tau_* \) and \( \chi_* \). 
The optimal values $\phi^*$ and $\phi_*$ are then given by $ \phi^* = \displaystyle\frac{\chi^*}{\tau^*} $ and $ \phi_* = \displaystyle\frac{\chi_*}{\tau_*} $.
Finally, to find the solution of \eqref{eq:optimization_problem}, we set \(\gamma^*_i = \displaystyle\frac{1}{\phi_{*i}}\) and \(\gamma_{*i} = \displaystyle\frac{1}{\phi^*_{i}}\) for $i\in V$.

\begin{remark} 
\label{rem:lower_upper_values}
The only constraints in our problem are the bounds imposed on the decision variables $\phi_i$. Consequently, the optimal solution of $C(\phi)$ must occur at either the lower or upper bound of each decision variable:
\bea 
\phi_i^*, \phi_{*i} \in \{\underline{\phi_i},\overline{\phi_i}\}, \quad \forall i\in V
\eea
Whether the optimal value corresponds to the lower or upper bound depends on the initial condition.
\end{remark}

% \begin{remark}
%     With the same LP problem, we can also determine the maximum value of $C(\phi)$.
% \end{remark}

% Knowing the consensus value is important, but also it would be nice to influence this consensus value using certain marketing campaign and have some information about what are we going to do to handle these. In the next section we focus on this problem.

Next,  we present a set of experimental results to highlight the practical usefulness of Theorem~\ref{th:lower_upper_consensus} and optimization problems \eqref{eq:min_opt}-\eqref{eq:opt_constraints}.

% Understanding the strategies necessary to guide consensus formation holds practical significance. The following section is dedicated to addressing this issue.

\subsection{Experimental results}
\label{sec:assump_differentiate}

% \begin{remark}
% Inequalities \eqref{eq:lower_deriv_condit} and \eqref{eq:upper_deriv_condit} are required to hold for most time instances rather than at every time step. This permits occasional violations, leading to a less restrictive condition compared to enforcing them uniformly. This will be analyzed more in depth in \ref{sec:assump_differentiate}
% \end{remark}

%To justify the relevance of Assumption~\ref{ass:bounded_projection_condition}, we present examples in which we empirically demonstrate that the assumption holds in the majority of cases.

To generate realistic network topologies, we employ Barabási–Albert (BA) graphs \cite{barabasi1999emergence}. Since BA graphs are inherently undirected, we convert them into directed graphs by randomly removing a subset of edges, while ensuring that the resulting graph remains strongly connected. We remove $20\%$ of the edges, a proportion that is sufficiently large to ensure a meaningful directed structure, while still maintaining strong connectivity of the graph.
We evaluate our approach across three scenarios:
\begin{itemize}
    \item Uniformly random $x(0)$ with the stubbornness model for $\gamma(k)$ given in~\eqref{eq:COCA_stubborness}.
    \item Uniformly random $x(0)$ with a uniform random model for~$\gamma(k)$.
    \item $\beta$-distribution of $x(0)$ with the stubbornness model for $\gamma(k)$ given in~\eqref{eq:COCA_stubborness}. 
\end{itemize}
Through these scenarios, we aim to encompass a range of realistic and practically relevant network topologies and initial opinions. For each scenario, we generate datasets consisting of 10,000 random graphs, with the number of nodes uniformly sampled from the range  $[10, 100]$. 
%The implementation and simulations are carried out in \texttt{Python}.

\subsubsection*{Uniformly random $x(0)$ with the stubbornness model for $\gamma(k)$ given in~\eqref{eq:COCA_stubborness}}
The initial opinion of the agents, \(x(0)\), is uniformly randomly generated within the range \([0.1, 0.9]\), and the stubbornness model in~\eqref{eq:COCA_stubborness} is used. In \(98.26\%\) of the cases, the stubbornness model for \(\gamma(\cdot)\) satisfies the lower bound condition in \eqref{eq:lower_deriv_condit} of Assumption~\ref{ass:bounded_projection_condition}, while in \(98.30\%\) of the cases it satisfies the upper bound condition in \eqref{eq:upper_deriv_condit}. In $96.80\%$ of the cases, both \eqref{eq:lower_deriv_condit} and \eqref{eq:upper_deriv_condit} are satisfied, justifying our choice to impose Assumption~\ref{ass:bounded_projection_condition}. For one example of a network, the evolution of the states $x(k)$ and the corresponding bounds $\alpha_{\min}$, and $\alpha_{\max}$ can be observed in Fig.~\ref{fig:one_dynamics}. Initial opinions are uniformly randomly sampled from the interval $x(0)\in [0.1,\, 0.9]^n$, and the consensus lies between the bounds.
\begin{figure}
    \centering
    \includegraphics[width=0.5\linewidth]{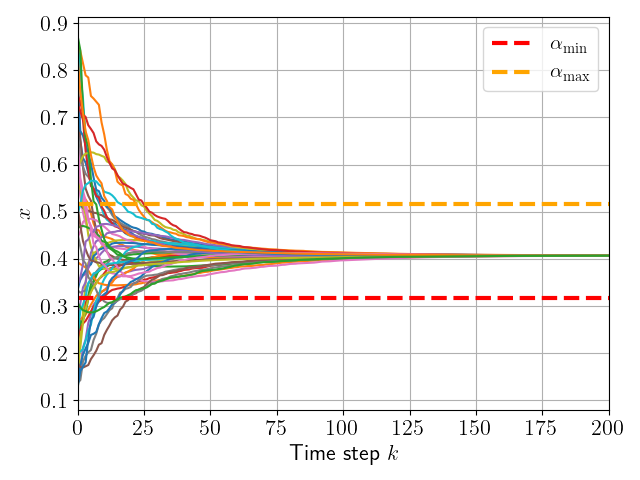}
    \caption{Evolution of $\underline{\theta}(k)$ and $\overline{\theta}(k)$}
    \label{fig:one_dynamics}
\end{figure}
We further observe that in all the experiments in which Assumption~\ref{ass:bounded_projection_condition} is not satisfied, the resulting consensus values obtained by solving \eqref{eq:opt_problem_proof} and \eqref{eq:opt_constr_proof} still lie within the prescribed~limits.

\subsubsection*{Uniformly random $x(0)$ with a uniform random model for~$\gamma(k)$}

% We consider BA graphs generated by randomly removing \(20\%\) of the edges from the original undirected BA graphs.

%Using randomly generated realizations of \( \gamma_i(k) \) for \( i = 1, \ldots, n \) and \( k = 1, \ldots, N^{\infty} \), the inequality~\eqref{eq:lower_deriv_condit} is satisfied in \( 83\% \) of the cases, while~\eqref{eq:upper_deriv_condit} holds in \( 83.58\% \) of the cases. Both conditions are simultaneously satisfied in \( 73.69\% \) of the cases. While there is a $26.31\%$ of the cases when Assumption \ref{ass:bounded_projection_condition} is not satisfied, the consensus value is still stays within these bounds in all the cases that we considered in the simulation.

The initial opinion of the agents, \(x(0)\), is uniformly randomly generated within the range \([0.1, 0.9]\), and we uniformly randomly generate the values of $\gamma_i(k)\in[\underline\gamma_i, \, \overline{\gamma}_i]$ for $i \in V$ and for every time~step $k$ until consensus is reached. Using randomly generated realizations of \(\gamma_i(k)\) for \(i \in V\), \eqref{eq:lower_deriv_condit} is satisfied in \(86.33\%\) of the cases, while~\eqref{eq:upper_deriv_condit} holds in \(85.61\%\) of the cases. Both conditions are simultaneously satisfied in \(76.71\%\) of the cases. As before, the consensus value remains within the computed bounds in all the simulations.

\subsubsection*{$\beta$-distribution of $x(0)$ with the stubbornness model for $\gamma(k)$ given in \eqref{eq:COCA_stubborness}}
\label{subsec:Cons_bound_scen3}
The initial opinion of the agents is sampled from a $\beta$-distribution \cite{johnson1995continuous} with parameters $\beta_{a}=2$ and $\beta_{b} = 5$ that is biased towards $0$ opinion on the range \([0.1, 0.9]\).
%An example of such a distribution is given in Fig. \ref{fig:beta_distrib_example}. 
Again we consider the stubbornness function in~\eqref{eq:COCA_stubborness}.
Given this distribution of $x(0)$, \eqref{eq:lower_deriv_condit} is satisfied in $99\%$ of the cases, while \eqref{eq:upper_deriv_condit} holds in $99\%$ of the cases. Both conditions are simultaneously satisfied in $98.18\%$ of the simulations. 

The distribution of $\alpha_{\min}-\alpha$ and $\alpha_{\max}-\alpha$ can be observed in Fig.~\ref{fig:beta_distrib_biased}. For all cases $\alpha_{\min}-\alpha\leq 0$ and $\alpha_{\max}-\alpha\geq 0$, i.e. the bounds are correct. We can observe that since $x(0)$ is biased towards $0$, the variance of $\alpha_{\min}-\alpha$ is smaller than that of $\alpha_{\max}-\alpha$.

% \begin{figure}
%     \centering
%     \includegraphics[width=0.99\linewidth]{img/beta_distrib_100.png}
%     \caption{A set of opinions samples from beta-distribution}
%     \label{fig:beta_distrib_example}
% \end{figure}

% \subsubsection{$\alpha$-distribution of $x(0)\in[0.2\, 0.8]$ with stubbornness model}
% We have the same setup as in the previous case, with the difference that now the range of the initial opinions is changed to $x(0)\in[0.2,\, 0.8]$. In this case both conditions of Assumption~\ref{ass:bounded_projection_condition} are satisfied in $98.71\%$ of the cases.

Fig.~\ref{fig:random_vs_stubborn_gamma} shows the distribution of the difference between the lower and upper bounds for all $3$ scenarios. %, see Fig.~\ref{fig:random_vs_stubborn_gamma}.
% \begin{figure}
%     \centering
%     \includegraphics[width=0.99\linewidth]{img/beta_range_distribution_2025_06_16.png}
%     \caption{Distribution of $\alpha_{\max}-\alpha_{\min}$ for both random and stubborn $\gamma$}
%     \label{fig:random_vs_stubborn_gamma}
% \end{figure}
In most cases, the estimated bounds are close to each other, with a mean difference of $0.19$ in the first two scenarios. In the third scenario, this mean difference is further reduced to $0.1$, indicating a more concentrated consensus~range. The bounds obtained using Theorem~\ref{th:lower_upper_consensus} have an average range of $0.19$ for Scenarios 1–2 and $0.1$ for Scenario 3, considerably narrower than with the method in \cite{ren2005consensus, chowdhury2016continuous}, which uses $\max(x(0))-\min(x(0))$, in our case lead to a range of $0.8$.
%Moreover, when the range of the initial opinions is reduced, the resulting consensus estimate also narrows, with a mean difference of $0.067$, further highlighting the tightness and effectiveness of the proposed approximation.

% \begin{figure}
%     \centering
%     \includegraphics[width=0.99\linewidth]{img/beta_bound_vs_final_opinion_beta_distrib_2025_07_22.png}
%     \caption{Distributions of $\alpha_{\min}-\alpha$ and $\alpha_{\max}-\alpha$ for biased $x(0)$}
%     \label{fig:beta_distrib_biased}
% \end{figure}

% \begin{figure}
%     \centering
%     \includegraphics[width=0.99\linewidth]{img/beta_range_distribution_with_all_2025_07_22.png}
%     \caption{Histogram of $\alpha_{\max}-\alpha_{\min}$}
%     \label{fig:random_vs_stubborn_gamma}
% \end{figure}

\begin{figure}[htbp]
    \centering
    % Left figure
    \begin{minipage}[b]{0.49\linewidth}
        \centering
        \includegraphics[width=\linewidth]{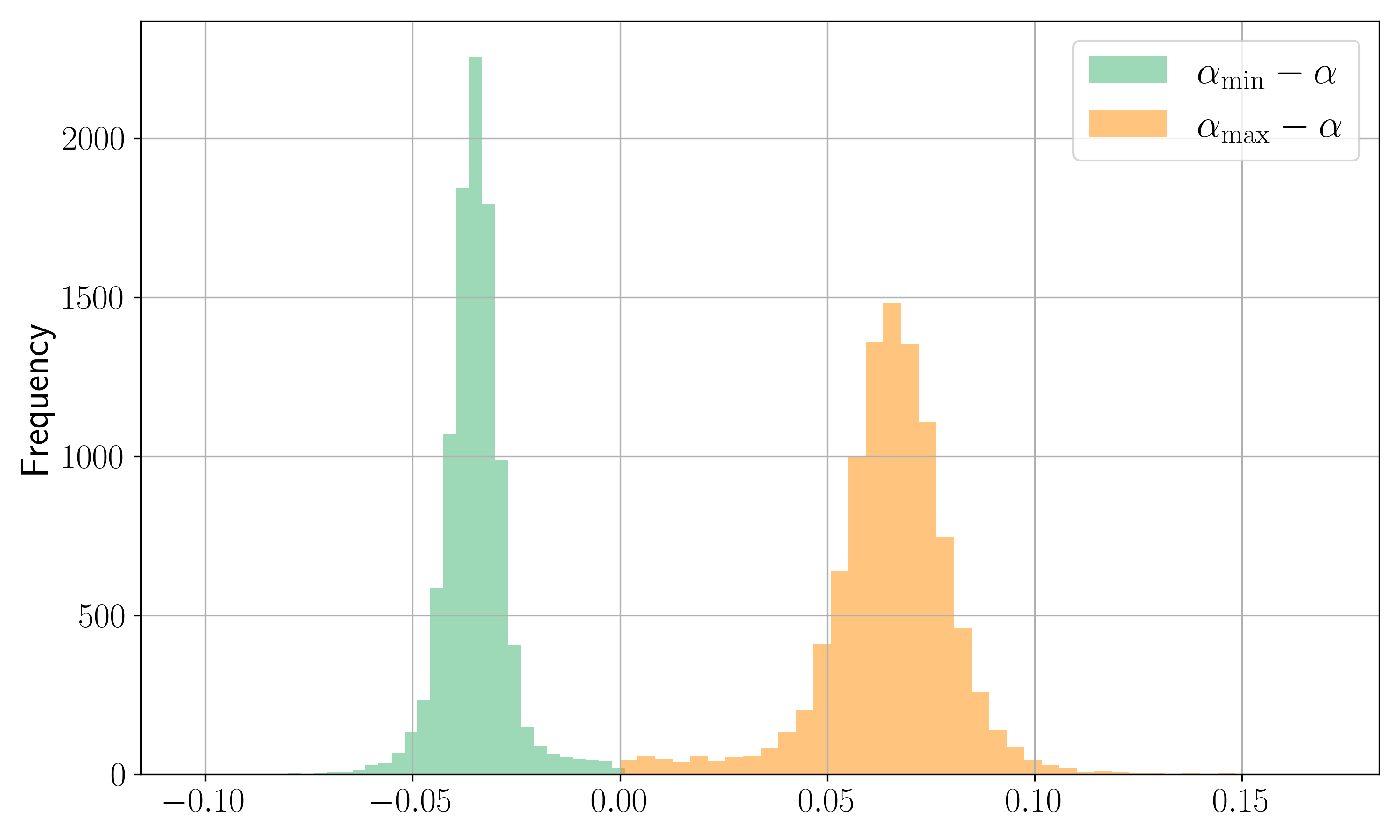}
        \caption{Distributions of $\alpha_{\min}-\alpha$, $\alpha_{\max}-\alpha$ for biased $x(0)$.}
        \label{fig:beta_distrib_biased}
    \end{minipage}%
    \hfill
    % Right figure
    \begin{minipage}[b]{0.49\linewidth}
        \centering
        \includegraphics[width=\linewidth]{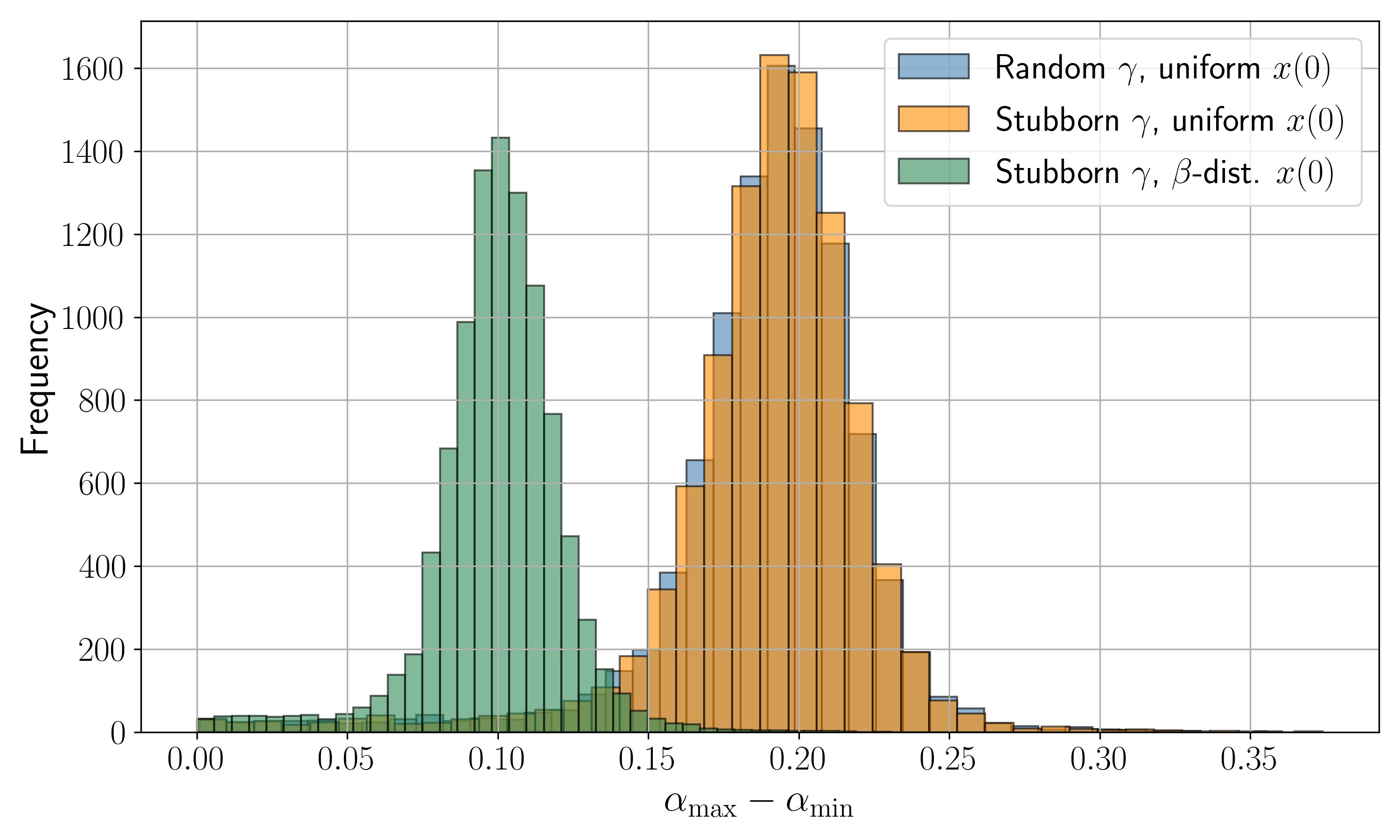}
        \caption{Histogram of $\alpha_{\max}-\alpha_{\min}$.}
        \label{fig:random_vs_stubborn_gamma}
    \end{minipage}
\end{figure}

\section{Consensus Control}
\label{sec:consensus_control}

Next, we will exploit the previously obtained consensus bounds to develop strategies to control the agents' states. %Specifically, we present a strategy to efficiently allocate a control budget to the agents.
We begin with the method and analytical results in Section~\ref{sec:control_analytical}, then continue with the experimental setup in Section~\ref{sec:control_experimental_setup}, and finally provide the experimental results in Section~\ref{sec:experimental_results}.

\subsection{Analytical results}
\label{sec:control_analytical}

The control input often shapes agents’ initial opinions quickly, since direct interventions such as targeted marketing result in rapid responses. In contrast, network-driven changes evolve gradually \cite{jia2015opinion, moruarescu2020space}.
%This is motivated by the fact that the interval between successive time steps is typically large relative to the timescale of the agents' intrinsic opinion dynamics, allowing the control to take immediate effect.
% The assumption that the marketing campaign instantly alters the agents' initial opinions is motivated by the fact that the time interval between two successive instances, $k$ and $k+1$, is usually sufficiently large relative to the timescale of the agents' intrinsic opinion dynamics.
Therefore, we consider the following control protocol adapted from~\cite{moruarescu2020space}:
\bea
x_u(0) &= d\,u + \mathrm{diag}(x(0))(\mathbf{1}-u),\\
u_i& \in [0, \overline{u}], \quad \forall i\in V,\\
\mathrm{with}\quad 0&\leq \overline{u} \leq 1,\,
\sum_{i=1}^n u_i\leq B
\label{eq:control_p0}
\eea
Here $x(0)$ is the initial opinion, $x_u(0)$ denotes the opinions of the agents after being instantaneously influenced by marketing campaign, $d \in \{0, 1\}$ represents the desired consensus value, $u = [ u_1 \,, u_2 \,,...,\, u_n ]^T$ is the control input vector, $\overline{u}$ indicates the maximum allowable value of any $u_i$ for $i\in V$, and $B$ denotes the total available budget of the marketing campaign. Therefore, the campaign occurs once, at time step $0$, following which we have the usual uncertain opinion dynamics:
\bea
x_u(k+1) =&\big(I+\mathrm{diag}(\gamma(k)L\big)x_u(k)
% x_u(0) =  &d\,u + \mathrm{diag}(x(0))(\mathbf{1}-u).
\label{eq:dynamics_with_control}
\eea
Our objective is to determine an allocation strategy of the control budget  that minimizes $|d - \alpha|$, which means we aim to bring the consensus value $\alpha$ as close as possible to the desired target $d$. However, since we do not have the exact value of  $\alpha$, we instead exploit our bounds to minimize $d - \alpha_{\min}$ (i.e., maximize $\alpha_{\min}$) when $d = 1$, or respectively minimize $ \alpha_{\max}-d$ (i.e., minimize $\alpha_{\max}$) when $d = 0$. To do this, we formulate the following theorem.
%Under Assumption \ref{ass:bounded_projection_condition}, we can allocate the marketing budget so that we maximize $\underset{u}{\max}\, \alpha_{\min}$ for $d=1$ and $\underset{u}{\min}\, \alpha_{\max}$ for $d=0$, since we have $\alpha_{\min}\leq\alpha \leq d$ for $d=1$ and $\alpha_{\max}\geq\alpha \geq d$ for $d=0$. This result is formalized in the following theorem.
\begin{theorem}
\label{th:control_nonlin_alloc}
Consider the objective function:
\bea
J(u,\gamma, d):= 
\frac{\nu^T\mathrm{diag}\big(\frac{1}{\gamma}\big)}{\nu^T\mathrm{diag}\big(\frac{1}{\gamma}\big)\mathbf{1}}\big(d\,u + \mathrm{diag}(x(0))(\mathbf{1}-u)\big)
\eea
Under Assumption~\ref{ass:bounded_projection_condition}, the solutions to the problems $\underset{u}{\max}\, \alpha_{\min}$ for $d=1$, and $\underset{u}{\min}\,\alpha_{\max}$ for $d=0$, are given by solving: 

%Under Assumption \ref{ass:bounded_projection_condition} the budget $B$ is optimally allocated among the agents of the network for the opinion dynamics model in \eqref{eq:dynamics_with_control} with the objective to maximize $\alpha_{\min}$ for $d= 1$, and minimize $\alpha_{\max}$ for $d=0$. To achieve these objectives we have the following functions:
\bea
d=1:\quad \underset{u}{\max}&\Big( \underset{\gamma}{\min} \,J(u,\gamma, 1)\Big) %:= \\&\frac{\nu^T\mathrm{diag}\big(\frac{1}{\gamma}\big)}{\nu^T\mathrm{diag}\big(\frac{1}{\gamma}\big)\mathbf{1}}\big(d\,u + \mathrm{diag}(x(0))(\mathbf{1}-u)\big)
\label{eq:control_d_1}
\eea
% subject to \eqref{eq:control_constraints},
% and
\bea
d=0: \quad \underset{u}{\min}&\Big( \underset{\gamma}{\max} \,J(u,\gamma, 0)\Big) %:= \\&\frac{\nu^T\mathrm{diag}\big(\frac{1}{\gamma}\big)}{\nu^T\mathrm{diag}\big(\frac{1}{\gamma}\big)\mathbf{1}}\big(d\,u + \mathrm{diag}(x(0))(\mathbf{1}-u)\big)
\label{eq:control_d_0}
\eea
both subject to: 
%\eqref{eq:control_constraints}, where the constraints are:
\bea
\gamma_i&\in [\underline{\gamma}_i,\, \overline{\gamma}_i],\quad 
u_i \in [0, \overline{u}], \quad \forall i\in V, \quad 
\sum_{i=1}^n u_i\leq B
\label{eq:control_constraints}
\eea

\end{theorem}

\begin{proof}

Consider $d=1$ first. Theorem~\ref{th:lower_upper_consensus} provides the lower and upper bounds $\alpha_{\min}$ and $\alpha_{\max}$ for an arbitrary initial condition $x(0)$. 
%As formulated in \eqref{eq:optimization_problem} to obtain the lower bound we have: 
% \beas
% \alpha_{\min} = \underset{\gamma}{\min}\quad C(\gamma) := \frac{\nu^T \mathrm{diag}\big(\frac{1}{\gamma}\big) x(0)}{\nu^T \mathrm{diag}\big(\frac{1}{\gamma}\big)}.
% \eeas
% From here we obtain $\gamma_*$ which gives $\alpha_{\min} = C(\gamma_*)$. 
By applying Theorem~\ref{th:lower_upper_consensus} and \eqref{eq:optimization_problem} to find the lower bound $\alpha_{\min}$ for the particular initial condition $x_u(0)$ modified by $u$, we get:
%When $d = 1$, applying Theorem~\ref{th:lower_upper_consensus} with \eqref{eq:optimization_problem}, we obtain the lower bound corresponding to the controlled initial condition $x_u(0)$, which is modified by $u$:
% Similarly, using \eqref{eq:dynamics_with_control} with $d=1$ we have:
\beas
\underset{\gamma}{\min}\quad  \frac{\nu^T \mathrm{diag}\big(\frac{1}{\gamma}\big) x_u(0)}{\nu^T \mathrm{diag}\big(\frac{1}{\gamma}\big)}
\eeas
Our objective is to find a control input $u$ that maximizes $\alpha_{\min}$. 
Using \eqref{eq:control_p0} to write $x_u(0)$ explicitly, this outer optimization problem writes as
\beas
\underset{u}{\max}&\bpm\underset{\gamma}{\min}\, \frac{\nu^T \mathrm{diag}\big(\frac{1}{\gamma}\big) \big(d\,u + \mathrm{diag}(x(0))(\mathbf{1}-u)\big)}{\nu^T \mathrm{diag}\big(\frac{1}{\gamma}\big)}\epm = \underset{u}{\max}\big(\underset{\gamma}{\min}\, J(u,\gamma, 1)\big) 
\eeas
Symmetrically, we obtain \eqref{eq:control_d_0} for $d=0$.
\end{proof}
% Theorem~\ref{th:control_nonlin_alloc} presents a nonlinear optimization problem, where one decision variable is maximized and another is minimized. 
Due to the presence of coupled terms, such as $\displaystyle\frac{1}{\gamma_i}u_i$, finding optimal solutions to the problems in Theorem~\ref{th:control_nonlin_alloc} is computationally challenging; therefore, in the following corollary, we introduce a computationally efficient approximate~solution. 
%By applying a change of variables, we reformulate the objective function as a linear program.
    
\begin{corollary}
\label{cor:lp_formulat}

Consider the objective function:
\bea
\tilde{J}(\tilde{u},\tilde{\phi}, d) := & \sum_{i=1}^n\nu_i \tilde{\phi}_i x_i(0) + \nu_i\tilde{u}_i (d-x_i(0))
\eea
where:
% The solutions 
% gives near optimal solutions to maximize $\alpha_{\min}$ for $d=1$ and minimize $\alpha_{\max}$ for $d=0$ under the control influence of $u$.
% Under Assumption~\ref{ass:bounded_projection_condition} 
\bea
\phi_i = \frac{1}{\gamma_i}, \quad
\tilde{\phi}_i := \phi_i\chi, \quad \tilde{u}_i := \phi_iu_i\chi,\quad 
\chi := \frac{1}{\nu^T\phi}
\label{eq:change_of_vars}
\eea
for all $i=1,..,n$. Relaxed solutions to the problems $\underset{u}{\max}\, \alpha_{\min}$ for $d=1$, and $\underset{u}{\min}\,\alpha_{\max}$ for $d=0$, are  given respectively by:
\bea
d=1: \quad \underset{\tilde{u},\tilde{\phi},\chi}{\max}\quad \tilde{J}(\tilde{u},\tilde{\phi},1)
\label{eq:obj_fun_control_lp_d_1}
\eea
% subject to \eqref{eq:obj_fun_lin_constr}, and
\bea
d=0: \quad \underset{\tilde{u},\tilde{\phi}, \chi}{\min}\quad \tilde{J}(\tilde{u},\tilde{\phi},0) 
\label{eq:obj_fun_control_lp_d_0}
\eea
both subject to %\eqref{eq:obj_fun_lin_constr}, where the constraints are:
\bea
\nu^T \tilde{\phi}&=1, \quad
A_{\phi}\tilde{\phi}\leq b_{\phi}\chi, \quad
A_u\tilde{u}  &\leq b_u \chi
\label{eq:obj_fun_lin_constr}
\eea
\bea
\sum_{i=1}^{n}\frac{\tilde{u_i}}{\underline{\phi}_i} \leq B \chi
\label{eq:obj_fun_nl_constr}
\eea
where 
\bea
A_{\phi} =& \bbm I\\-I
\ebm, \quad 
b_{\phi} = \bbm \overline{\phi}_1 & \cdots & \overline{\phi}_n & -\underline{\phi}_1 & \cdots & -\underline{\phi}_n  \ebm^T,\quad
A_u = I, \quad b_u = \bbm \overline{u}\\...\\ \overline{u}
\ebm
% A_{B} = &\mathbf{1}^T, \quad b_{B} = B
\eea 
%The obtained objective functions are linear.

\end{corollary}

\begin{proof}
Using the change of variables in \eqref{eq:change_of_vars}, the nonlinear optimization problems in \eqref{eq:control_d_1} and \eqref{eq:control_d_0} become linear. The budget constraint $\sum_{i=1}^n u_i\leq B$ is rewritten as
\beas
\sum_{i=1}^{n}\frac{\tilde{u_i}}{\tilde{\phi}_i} \leq B 
\eeas
which is transformed as follows
\beas
\sum_{i=1}^n \frac{\tilde{u}_i}{\tilde{\phi}_i} = \sum_{i=1}^n \frac{\tilde{u}_i}{\phi_i\chi}&\leq B\,\Longleftrightarrow \,
\sum_{i=1}^n \frac{\tilde{u}_i}{\phi_i}&\leq B\chi
\eeas
We know that $\phi_i\geq\underline{\phi}_i$, therefore:
\bea
\sum_{i=1}^n \frac{\tilde{u}_i}{\phi_i} \leq \sum_{i=1}^n \frac{\tilde{u}_i}{\underline{\phi}_i} &\leq B\chi
\label{eq:relaxed_constraint}
\eea
which is a linear constraint.
\end{proof}
\color{black}
The problem of Corollary~\ref{cor:lp_formulat} can be efficiently solved using linear programming.
It is a relaxation in the sense that, while the budget is not exceeded, it may not be fully utilized due to the relaxed constraint in \eqref{eq:relaxed_constraint}.
% \end{remark}
% This can be interpreted as an update to the agents' initial opinions resulting from the marketing campaign. 
% Although, the exact consensus value cannot be determined directly, it can be bounded using Theorem~\ref{th:lower_upper_consensus}.
% %using the bounds   at computing an optimistic—i.e., potentially achievable—consensus value. % while simultaneously determining an optimal allocation of the available budget.
% We have the following objective~function:
% \bea
% J(u) = |\mathbf{1}d-\lim_{k\rightarrow \infty }x(k)| = |d-\alpha|
% \label{eq:opt_control}
% \eea
% Since the exact value of $ \alpha $ is unknown, we consider the bounds $ \alpha_{\min} $ and $ \alpha_{\max} $. This allows us to obtain an optimistic estimate of the potentially achievable consensus value. When \( d = 0 \), the optimistic scenario is
% \bea
% J^*_{d=0}(u) := |d - \alpha_{\min}|,
% \label{eq:opt_control_d_0}
% \eea
% whereas for \( d = 1 \), we have:
% \bea
% J^*_{d=1}(u) := |d - \alpha_{\max}|.
% \label{eq:opt_control_d_1}
% \eea

To illustrate the effectiveness of the allocation strategy defined in Corollary~\ref{cor:lp_formulat}, we next present a series of numerical examples, comparing to allocation strategies drawn from the literature. We begin by outlining the setup of the examples, then describe the baseline. Finally, we present and discuss the empirical results in~detail.

% To highlight the usefulness of the allocation strategy defined in Corollary~\ref{cor:lp_formulat} we present first the setup of the examples which is followed by the actual exampels
% following illustrative examples.

%-------------------------------------------------------------------------
\subsection{Experimental setup and baselines}
\label{sec:control_experimental_setup}

% To ensure comparability with existing results in \cite{moruarescu2020space},

The total available budget is given by \( B = n_b \, \overline{u} \), where \( n_b \) denotes the number of agents that can be influenced at the maximum control level \( \overline{u} \). We restrict the allocation to a binary decision: each agent receives either the maximum allowable control input or none at all. While Corollary~\ref{cor:lp_formulat} does not explicitly impose this binary restriction, the constraints in \eqref{eq:control_constraints} act on each decision variable individually, i.e., $\gamma_i \in [\underline{\gamma}_i, \,\overline{\gamma}_i]$ and $u_i\in [0,\, \overline{u}]$ for $i=1,...,n$. As a result, the optimal solution naturally occurs at either the minimum or maximum values of the decision variables, implicitly leading to binary inputs. Next, we consider three allocation strategies.
% \begin{itemize}
%     \item Brute-force allocation.%: Enumerate all possible allocation configurations and select the one that maximizes \eqref{eq:control_d_1}.
%     \item Baseline: Allocating resources based on the scheme in~\cite{moruarescu2020space}.
%     \item Allocating resources based on Corollary~\ref{cor:lp_formulat}.
% \end{itemize}

% 
% In what follows, we present these $3$ allocation methods.
%--------------------------------------------------
\subsubsection*{Brute-force allocation}

When the number of agents in the network is small, it is computationally feasible to enumerate all possible budget allocations and select one that maximizes \eqref{eq:control_d_1} or minimizes \eqref{eq:control_d_0} \cite{chen2011finite}.
%In our study, we considered networks with a maximum of \( n = 12 \) agents, for which exhaustive enumeration of all possible budget allocations remains computationally tractable.
%-------------------------------------------------------------
\subsubsection*{Linear-dynamics baseline}
\label{sec:baseline}

We exploit the linear case in \cite{moruarescu2020space}, where agents are ranked according to their centrality and initial opinion, combined into the \emph{influence power} defined as:
\bea
\label{eq:space_influence_power}
\rho_i := \nu_i |d - x_i(0)|
\eea
Note that the centralities are given by the normalized left eigenvector $\nu$ corresponding to the $0$ eigenvalue of $L$.
% As linear dynamics are considered in \cite{moruarescu2020space}, the term $\nu$ is constant. Since in our dynamics $\nu$ is dynamically changing as it is stated in \eqref{eq:left_eigen_v}, we consider three options for $\nu$: $\underline\nu$, $\overline{\nu}$, and $\nu_{\gamma(p(0))}$ to define the centrality of the agents.
%The use of $\nu$ in \eqref{eq:space_influence_power}can be adapted by replacing it with either \( \overline{\nu} \) or \( \underline{\nu} \), depending on \eqref{eq:control_d_0} or \eqref{eq:opt_control_d_1}.
Based on \eqref{eq:space_influence_power}, the optimal investment profile is given by Proposition 5.1 of \cite{moruarescu2020space}:
\bea
u_i =
\begin{cases} 
\overline{u} & \text{if } i \in \mathcal{I}_b \\
0 & \text{otherwise}
\end{cases}
\label{eq:agent_infl_power}
\eea
% \textcolor{blue}{
where \( \mathcal{I}_b \subseteq \{1, \dots, n\} \) denotes the set of \( n_b \) agents with the largest influence powers \( \rho_i \).
%whether an optimistic or pessimistic scenario is considered. %Moreover, in the context of a stubbornness model where \( \gamma(\cdot) \) depends on the opinion values, a specific value \( \nu_{p_0} \) can be derived based on \( x(0) \), i.e., \( \nu_{p_0} := \nu(\gamma(x(0))) \).
In the linear case of \cite{moruarescu2020space}, the term $\nu$ is constant. In contrast, here $\nu_{\gamma(k)}$ varies over time. As an approximation, we use $\nu_{\gamma(0)}$ to perform the allocation, with $\gamma(0)=\gamma(x(0))$ defined in \eqref{eq:COCA_stubborness}. While allocation \eqref{eq:agent_infl_power} is optimal in the linear case, in the nonlinear case considered in \eqref{eq:dynamics_with_control}, no such theoretical guarantees are available.
Nevertheless, once the budget allocation is determined,  Theorem~\ref{th:lower_upper_consensus} can still be used to bound the consensus value and assess its proximity to the desired opinion~\( d \).
% }
%This approach can be interpreted as a two-step procedure: first, the budget is allocated according to \eqref{eq:agent_infl_power}; second, the impact is measured using our consensus~bounds.

% This approach provides the optimal solution in the linear case, as presented in \cite{moruarescu2020space}. However, we lack an optimality guarantee when dealing with a nonlinear model containing uncertain terms, as in \eqref{eq:coca_model_matrix_form}. 

% In the following section, we propose an approach that aims to identify the globally optimistic outcome.
% the allocation problem in a single step, directly incorporating the effect of the control input on the consensus outcome.
% Due to the nonlinearity and the high sensitivity on the term \( \nu \) in this context, there are no theoretical guarantees regarding the optimality of the resulting allocation strategy.

\subsubsection*{Corollary~\ref{cor:lp_formulat}-based allocation}

% Our main approach is to allocate the budget  Corollary~\ref{cor:lp_formulat} offers a relaxed budget allocation strategy. 
% Our main approach is to allocate the budget even 3
% \textcolor{blue}{
Finally, we apply our proposed method based on Corollary~\ref{cor:lp_formulat}. Note that the linear-dynamics baseline works in two separate stages: the budget is first allocated, and only then are the lower bounds determined. Corollary~\ref{cor:lp_formulat}, by contrast, unifies these steps into a single, theoretically grounded procedure in which the budget is allocated explicitly to maximize the resulting lower bound.
As discussed in Section~\ref{sec:control_analytical}, after applying Corollary~\ref{cor:lp_formulat} once, it is possible that part of the budget remains unused. To guarantee full utilization, the remaining budget is redistributed among the agents that have not yet received the maximum allocation $\bar{u}$, by repeatedly applying Corollary~\ref{cor:lp_formulat} until the budget is fully exhausted. In the end, the entire budget is used, and each funded agent attains the maximum $\bar{u}$.
% }
% While relaxing the nonlinear constraints to obtain a linear formulation, we encounter the issue that the budget $B$ may not be fully used. To address this, we include the following two additional allocation strategies:
% \begin{itemize}
%     \item If an agent is assigned any non-zero percentage of the maximum control input $\overline{u}$, it will receive the full value~$\overline{u}$.
%     \item If there is any remaining budget after applying the allocation strategy described above, we reapply Corollary~\ref{cor:lp_formulat} to the agents that have not yet received any control input.
% \end{itemize}
 
% One potential drawback of this approach is that the budget \( B \) may not be fully used. To address this issue, any remaining budget can be redistributed through an additional allocation.

% \begin{remark}
% The same methodology can be employed to compute an optimistic consensus value in the case where \( d = 1 \), using \eqref{eq:opt_control_d_1}.
% \end{remark}

\subsection{Experimental results}
\label{sec:experimental_results}
We present a set of illustrative examples using the above presented three allocation strategies. For simplicity, we consider $d = 1$ in all examples, as the same approach can be applied symmetrically for $d = 0$. 

In the first example, we consider a small-scale network to analyze the linear-dynamics baseline and the Corollary~\ref{cor:lp_formulat}-based allocations in comparison to the brute-force allocation. Next, we examine a large-scale network with $510$ nodes, considering different sets of initial opinions to explore a variety of scenarios. In this large-scale network the brute-force allocation is not computationally feasible.

% In the first example, we examine a small-scale network, aiming to determine the with respect to the brute-force allocation. This is then followed by a large-scale network, where the goal is to observe different initial opinions and see the allocation based on the baselin or the Corollary~\ref{cor:lp_formulat}.

% the allocation in a real, large-scale scenario

\subsubsection{Consensus control in small-scale networks}

We first focus on maximizing \eqref{eq:control_d_1}, for a directed graph with $12$ agents with the network topology in Fig.~\ref{fig:network_topology}. We generate $1,000$ initial conditions where each agent's state is uniformly randomly distributed in the interval $[0.1,\,0.9]$. We set $\overline{u}=0.2$ and $n_b=3$, so three agents get a budget $\overline{u}$. The values for $\gamma(k)$ in \eqref{eq:bounded_nonlin} are sampled from a uniform random distribution with a range given by $\underline{\omega}=0.03$ and $\overline{\omega}=0.25$. Table~\ref{tab:example_table} presents the average values of $\alpha_{\min}$, denoted by $\bar{\alpha}_{\min}$. The column $\alpha_{\min}\%$ presents the ratio between the linear-case or Corollary~\ref{cor:lp_formulat} based allocation with respect to the brute-force allocation.
We can observe that results using Corollary~\ref{cor:lp_formulat} $\alpha_{\min}$ are almost $6$ points larger than the linear baseline allocation.
%We can observe that, the Corollary~\ref{cor:lp_formulat}-based allocation strategy results in a higher consensus value even above the brute-force allocation, but that is due to the randomness

% \begin{table}[ht]
% \centering
% \begin{tabular}{|l|c|c|c|c|}
% \hline
% \textbf{Strategy} & $\bar{\alpha}$ & $\bar{\alpha}_{\min}$ & $\alpha\%$ & $\alpha_{\min}\%$ \\
% \hline
% Baseline with $\nu_{\gamma(0)}$ & $0.536$ & $0.445$ & $97.7\%$ & $93.13\%$ \\
% Corollary~\ref{cor:lp_formulat} & $0.551$ & $0.473$ & $100.4\%$ & $98.91\%$ \\
% Brute-force                     & $0.548$ & $0.478$ & $100\%$ & $100\%$ \\
% \hline
% \end{tabular}
% \caption{Comparison of budget allocation strategies}
% \label{tab:example_table}
% \end{table}

\begin{table}[ht]
\centering
\caption{Comparison of budget allocation strategies}
\label{tab:example_table}
\begin{tabular}{lcc}
\toprule
\textbf{Allocation strategy} & $\bar{\alpha}_{\min}$ & $\alpha_{\min}\%$ \\
\midrule
Baseline with $\nu_{\gamma(0)}$ & 0.445 & 93.13\% \\
Corollary~\ref{cor:lp_formulat} & 0.473 & 98.91\% \\
Brute-force & 0.478 & 100\% \\
\bottomrule
\end{tabular}
\end{table}

\subsubsection{Consensus control in large-scale networks}

% In this example we consider a realistic directed BA graph with $1,000$ agents, when a set of agents connected on a social network their opinion about climate change is modelled. Our goal is to highlight the importance of climate change for the agents in the network, therefore we apply a set of marketing campaigns. We consider for the desired consensus to be $d=1$, having the meaning that the agents knows and understands the importance of paying attention on the climate change. We sample the distribution of the initial opinions from a beta-distribution with parameters $\mathrm{beta}_{\alpha}$ and $\mathrm{beta}_{\alpha}$ having values on the interval of $\mathrm{beta}_{\alpha},\mathrm{beta}_{\alpha}\in[0,\,5]$. This covers a lot of possible initial opinions distribution cases. see Fig. y for some of the cases.

% We consider $\overline{u}=0.2$, and the max number of agents that can receive a this marketing campaing is $50$.

% Now we consider two allocation strategies: the baseline, and the Corollary~\ref{cor:lp_formulat} based allocation. In Fig x., we denote with $o$ the cases when the consensus obtained using the allocation strategy from Cor1 is above $0.5$, while with $.$ we denote those consensus values that are above $0.5$, and they are based on the baseline allocation strategy. 

% We can observe in Fig. x that the Cor 1 based allocation strategy is gives a wider range of cases, when the minimum consensus bound is higher than $0.5$ in comparsion to that of the baseline.

In this example, we consider a larger-scale directed BA graph consisting of $510$ agents. The initial opinions are sampled from $\beta$-distributions, considering  $\beta_{a}, \beta_{b} \in \{0.5,\,0.75,\,1,\,\ldots,\,3.5\}$, and all possible parameter pairs are considered. This range captures a wide variety of possible initial opinion distributions, see Fig.~\ref{fig:beta_distrib} for some examples.
% \begin{figure}
%     \centering
%     \includegraphics[width=0.99\linewidth]{img/beta_distributions_2025_07_25.png}
%     \caption{Examples of $\beta$-distribution}
%     \label{fig:beta_distrib}
% \end{figure}

\begin{figure}[htbp]
    \centering
    \begin{minipage}[b]{0.48\linewidth}
        \centering
        \includegraphics[width=\linewidth]{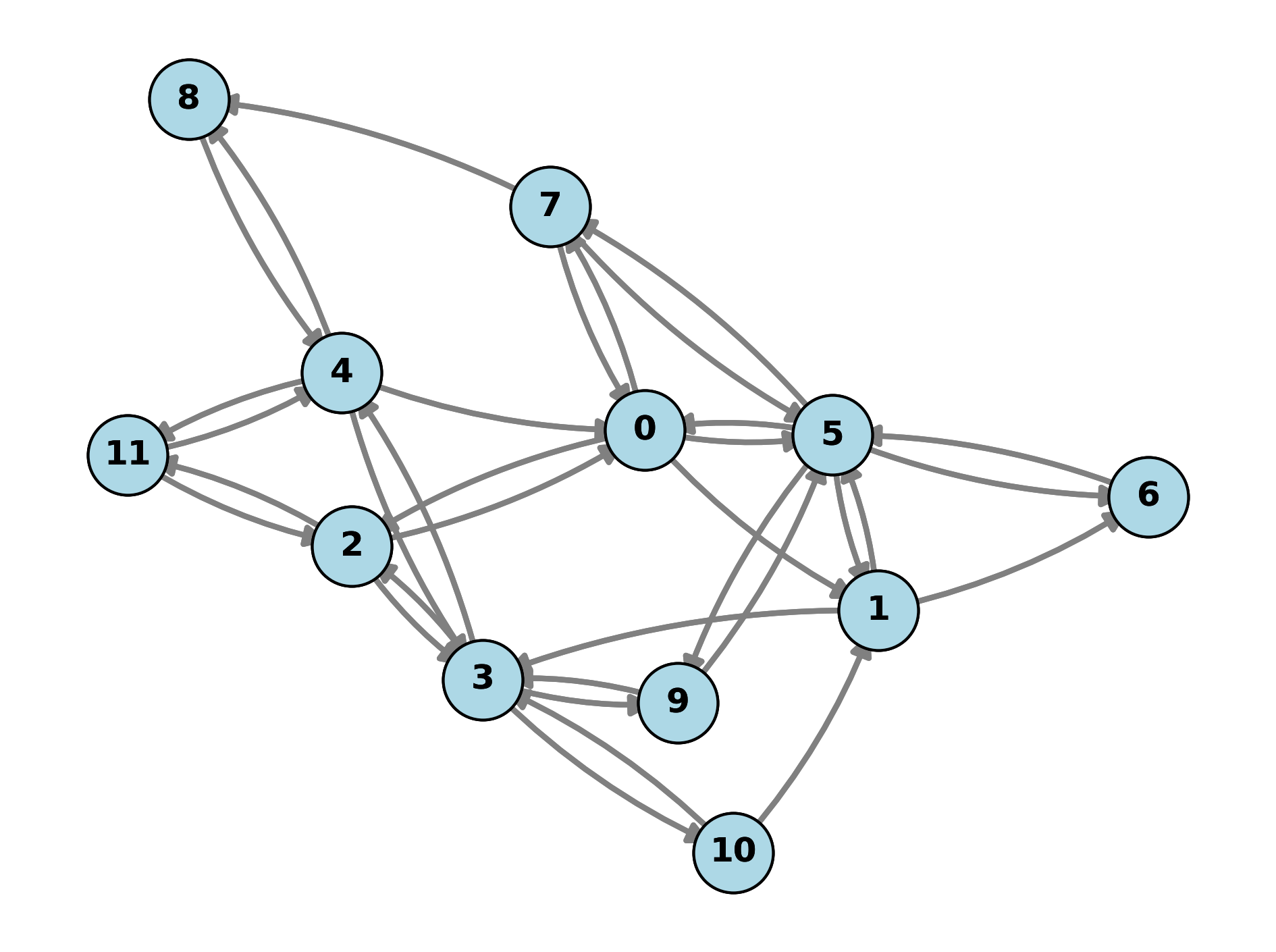}
        \caption{Network topology}
        \label{fig:network_topology}
    \end{minipage}\hfill
    \begin{minipage}[b]{0.48\linewidth}
        \centering
        \includegraphics[width=\linewidth]{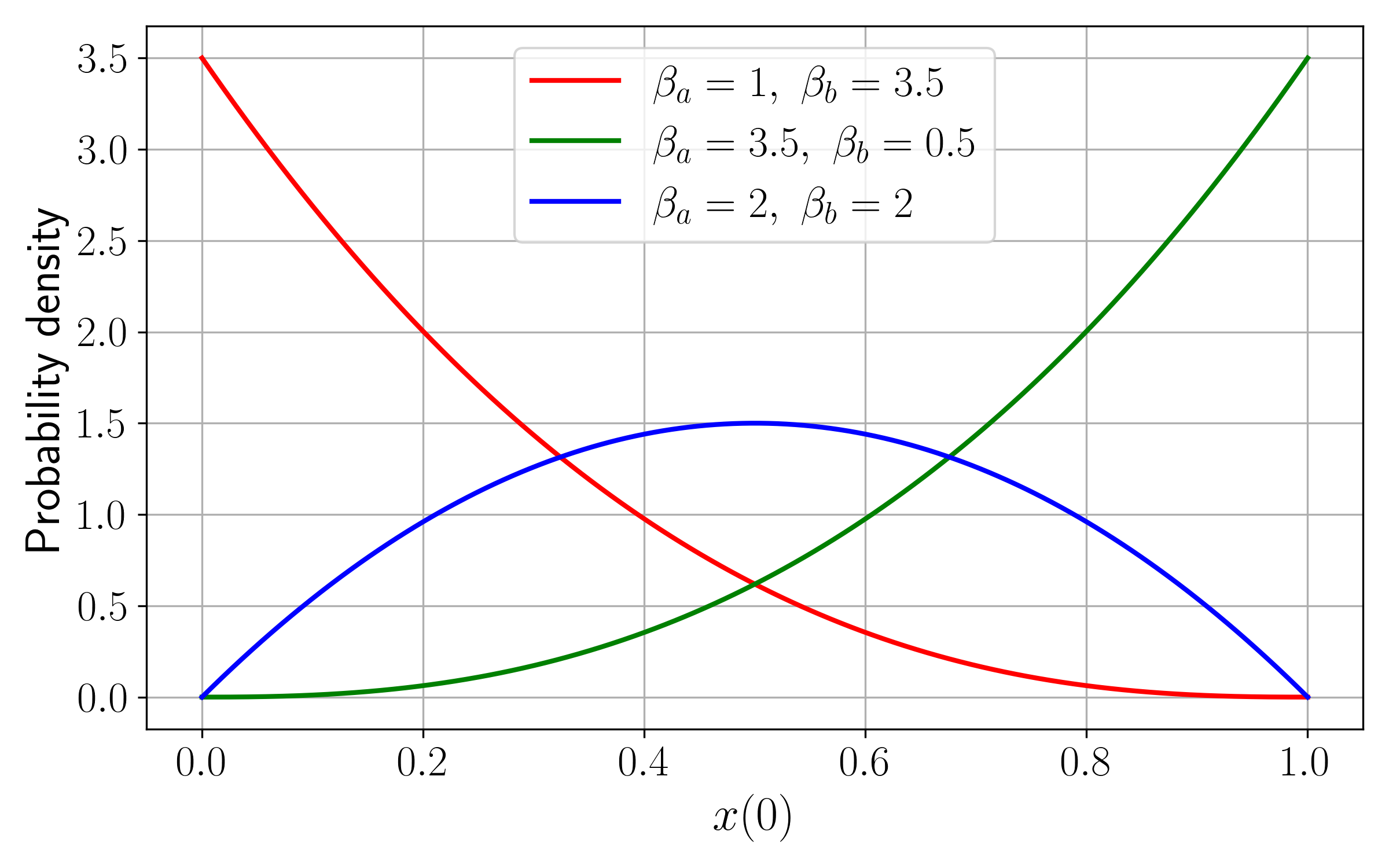}
        \caption{Examples of $\beta$-distributions}
        \label{fig:beta_distrib}
    \end{minipage}
\end{figure}

% \begin{figure}
%     \centering
%     \includegraphics[width=0.99\linewidth]{img/graph_case_1.png}
%     \caption{Network topology}
%     \label{fig:network_topology}
% \end{figure}

\begin{figure}[ht]
    \centering
    \includegraphics[width=0.52\linewidth]{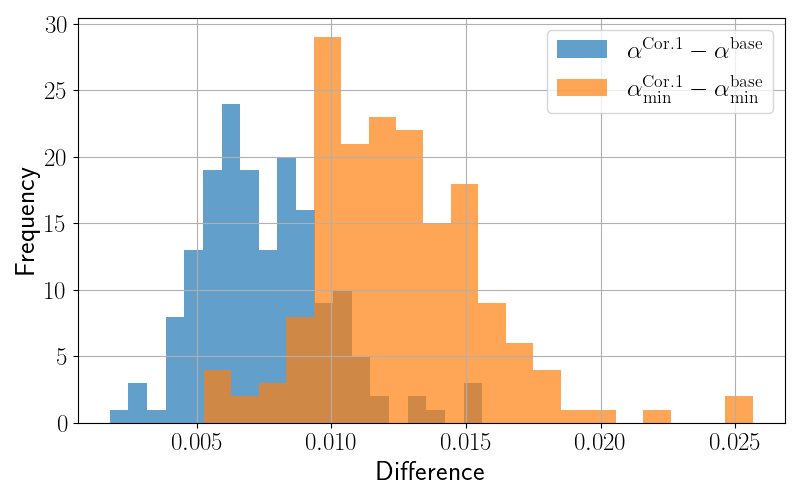}
    \caption{Differences of $\alpha$ and $\alpha_{\min}$ with Corollary~\ref{cor:lp_formulat} and baseline allocations}
    \label{fig:histogram_alloc_strategies}
\end{figure}

The maximum control input is $\overline{u} = 0.2$, and the maximum number of agents that receive marketing influence is $n_b=50$. Next, we examine two allocation strategies: baseline allocation and the allocation based on Corollary~\ref{cor:lp_formulat} (recall brute-force allocation is no longer feasible). Fig.~\ref{fig:histogram_alloc_strategies} shows the distribution of the differences $\alpha_{\min}^{\mathrm{Cor.1}} - \alpha_{\min}^{\mathrm{base}}$ and $\alpha^{\mathrm{Cor.1}} - \alpha^{\mathrm{base}}$, where  $\alpha_{\min}^{\mathrm{Cor.1}}$ and $\alpha_{\min}^{\mathrm{base}}$ denote the values of $\alpha_{\min}$ obtained using Corollary~\ref{cor:lp_formulat} and the baseline strategy, respectively. Similarly, $\alpha^{\mathrm{Cor.1}}$ and $\alpha^{\mathrm{base}}$ refer to the actual consensus values resulting by simulating to converge the network after applying the solution from Corollary~\ref{cor:lp_formulat} and the baseline. In all cases, the allocation based on Corollary~\ref{cor:lp_formulat} yields better values for $\alpha_{\min}$. 

% While the baseline allocation yields useful results, this can be further improved by Corollary~\ref{cor:lp_formulat} based allocation. 
An increase in $\alpha_{\min}$ by $0.02$ may appear marginal at first, but its impact in large-scale social networks can be substantial. For example, when the consensus value is interpreted as the proportion of agents leaning toward preference for a particular product, even a small shift can translate into meaningful business outcomes \cite{gelman2012probability}.

\section{Conclusions}
\label{sec:conclusions}

This paper presented a novel approach to approximate the consensus value for a class of nonlinear uncertain consensus dynamics. We began by introducing and motivating the model, followed by a set of analytical results for finding bounds on the consensus value.
%We conducted a series of numerical experiments that demonstrated the relevance and effectiveness of the proposed approach.
%%%
Building on these bounds, we examined a control application, 
where the goal was to optimally allocate a marketing budget among agents 
to minimize the gap between the consensus and a target value. 
We developed an optimization strategy that simultaneously determines both the control input and the bounds, which yields better results than the baseline approach from the literature.
% In particular, we allocate the marketing budget that minimizes the distance $d - \alpha_{\min}$ when $d = 1$, and $\alpha_{\max} - d$ when $d = 0$.
% We started with a brute-force allocation, followed by a baseline approach from the literature~\cite{moruarescu2020space}, 
% and finally, we refined the allocation using the linear programming-based method 
% presented in Corollary~\ref{cor:lp_formulat}.
%%%
% Building on these bounds, we extend the analysis to a control application, aiming to optimally allocate a given marketing budget among agents to minimize the distance between the consensus and the target value.

% We first examined the brute-force allocation as a direct strategy, then considered a baseline approach from the related literature \cite{moruarescu2020space}. Finally, we improved it using the linear programming–based method formulated in Corollary~\ref{cor:lp_formulat}.

% Specifically, we presented an approach to optimal allocation of a marketing budget with the objective of minimizing the distance $d-\alpha_{\min}$ for $d=1$, and $\alpha_{\max}-d$ for $d=0$.

% The optimization problems were formulated as linear programs, ensuring computational efficiency approaches. 

%Overall, the framework presented in this paper provides both theoretical insight and practical tools for consensus approximation and influence optimization in complex opinion networks.

% \section{Future directions}
% \label{sec:future_directions}
Future work could explore the specific COCA model further, with a focus on relaxing Assumption~\ref{ass:bounded_projection_condition}.
Another direction is to explore alternative opinion dynamics models, such as Hegselmann–Krause models \cite{rainer2002opinion} or continuous-opinion discrete-action models \cite{martins2008continuous}. It is also useful in practice to study settings with multiple marketing campaigns, where budgets must be allocated across both agents and time.

\section{Acknowledgement}
This work was supported by project DECIDE, no. 57/14.11.2022 funded
under the PNRR I8 scheme by the Romanian Ministry of Research, Innovation, and Digitisation.
%% Use \subsubsection, \paragraph, \subparagraph commands to 
%% start 3rd, 4th and 5th level sections.
%% Refer following link for more details.
%% https://en.wikibooks.org/wiki/LaTeX/Document_Structure#Sectioning_commands

\bibliographystyle{elsarticle-num}
\bibliography{z_opinion_dynamics}
%% For numbered reference style

%   \textit{\LaTeX: a document preparation system},
%   Addison Wesley, Massachusetts,

% \end{thebibliography}
\end{document}